\documentclass[11pt,a4paper]{article}
\pdfoutput=1
\usepackage{lmodern}
\usepackage[T1]{fontenc}
\usepackage[utf8]{inputenc}
\usepackage[margin=2.5cm]{geometry}
\usepackage{amsmath,amssymb,amsthm}
\usepackage{booktabs,tabularx,array}
\usepackage{graphicx}
\usepackage[hidelinks]{hyperref}
\usepackage[protrusion=false,expansion=false]{microtype}
\usepackage{setspace}
\usepackage{caption}
\usepackage{natbib}
\bibpunct{(}{)}{;}{a}{}{,}
\theoremstyle{plain}
\newtheorem{proposition}{Proposition}
\newtheorem{lemma}[proposition]{Lemma}
\newtheorem{corollary}[proposition]{Corollary}
\theoremstyle{definition}
\newtheorem{assumption}{Assumption}
\newtheorem{definition}[proposition]{Definition}
\theoremstyle{remark}
\newtheorem{remark}[proposition]{Remark}

\newcommand{\Prob}{\mathbb{P}}

\title{\bfseries Three Ceilings:\\[3pt]
\large Model Monoculture, Solvency, and the Penalty Doctrine\\
in Markets for Expert Services}

\author{Andreas Bauer\thanks{Aegis Compliance and Strategies O\"U, Tallinn, Estonia.
Correspondence: \texttt{a.bauer@science.us.org}. ORCID: 0000-0002-6539-9353.
\emph{Declarations of interest:} the author is Managing Director and Founder of Aegis Compliance
and Strategies O\"U, a compliance and strategy advisory firm, and is preparing a practitioner
book that draws on this line of work; both are disclosed in the interest of transparency.
\emph{Data availability:} replication code, the full parameter-provenance table and the random
seeds for every reported number are openly available; see Appendix~\ref{app:data}.
This paper consolidates and in one respect corrects three companion working papers by the
author, Bauer (2026a,b,c); Appendix~\ref{app:prov} records the provenance of every component.
All errors are my own.}\\[2pt]
\normalsize Aegis Compliance and Strategies O\"U}
\date{Working Paper --- this version: August 2026}

\begin{document}
\maketitle
\vspace{-0.8cm}

\begin{abstract}\noindent
When generative AI drives the marginal cost of a persuasive expert artefact toward zero,
production-cost signals of competence collapse and outcome-contingent liability commitments take
their place. Such commitments look robust to better AI: a positive failure rate always leaves a
residual to price. We show that this robustness rests on an unstated homogeneity assumption, and
that a liability commitment faces three ceilings, not one. AI error decomposes into idiosyncratic
and common components; capability growth eliminates the idiosyncratic part faster
\citep{kim2025}, so the residual becomes progressively common --- and common error is precisely
what a verifier drawn from the same foundation model cannot observe: judge scores correlate with
judge--generator model similarity at an average $r=0.84$ \citep{goel2025}. Provability is therefore
state-dependent, $\theta_{\mathrm{eff}}=\theta_0(1-\xi\kappa)$. Separation requires the
commitment that must be posted to fall below what can be posted,
$v/\theta_{\mathrm{eff}} \le \min\{\bar L,\, mv\}$, yielding the survival condition
$\xi\bar\kappa \le 1-1/(M\theta_0)$ with $M=\min\{\bar L/v,\,m\}$ and a regime switch at the
ticket size $v^{*}=\bar L/m$: small engagements are constrained by the penalty doctrine, large
ones by capital. Calibrated to six engagement types across four jurisdictions, the generative-AI
insurance exclusions effective January 2026 shift eight of twenty-four cells from doctrine-bound
to solvency-bound and cut the cells surviving $\xi=0.4$ from twelve to four. The insurance
withdrawal is a civil-law event; the common-law cells were already at the boundary.
\medskip


\end{abstract}

\newpage
\onehalfspacing

\section{Introduction}

Markets for expert services are credence-goods markets: the buyer cannot assess quality even
after consumption \citep{darby1973,dulleck2006}. One historical response was signalling. Because
a rigorous legal memorandum or a clean codebase was expensive to produce and more expensive for a
low-competence provider to imitate, the artefact carried information in the sense of
\citet{spence1973} and \citet{riley1979}.

Generative AI removes this. When the marginal cost of a persuasive artefact approaches zero for
every type, the single-crossing property fails and the production-cost signal collapses. What
survives is an \emph{outcome-contingent} commitment --- a liability cap, an indemnity, a
liquidated-damages clause --- whose expected cost depends on the residual states in which the AI
fails and human competence must rescue the engagement. This is the warranty-as-signal mechanism
of \citet{grossman1981}, \citet{lutz1989} and \citet{daughety1995}. Its apparent attraction is
immunity to capability growth: as long as the failure probability $\pi$ is positive, something is
left to price, and the signal attenuates rather than vanishes.

This paper makes two corrections, one to the mechanism and one to the framing.

\paragraph{The mechanism.} The attenuation argument treats the residual as homogeneous. AI error
decomposes into an idiosyncratic component and a common component arising from shared training
corpora, shared architectural priors and shared blind spots. \citet{kim2025}, evaluating over 350
language models, document substantial error correlation --- on one leaderboard dataset models
agree 60\% of the time when both err --- and, decisively, that larger and more accurate models
have highly correlated errors even with distinct architectures and providers. Capability growth
eliminates the random errors first; what remains is increasingly shared.

That matters because enforcing an outcome-contingent commitment requires that failure be
\emph{provable}. If the reviewer is drawn from the same foundation model as the provider's
generator, then in exactly the states where error is common the reviewer shares the blind spot.
\citet{goel2025} measure this directly: across LLM-as-judge settings they find a positive
correlation, averaging Pearson $r=0.84$ across nine judges, between judge scores and
judge--generator model similarity (their \S3.2, Fig.~3), and conclude
that using a held-out model as judge is insufficient because similarity is a confounder that must
be controlled. \citet{denisov2026} show the complement: aggregating more samples cannot
substitute for a verifier when errors are correlated, and a negative control on random strings
yields persistent above-chance agreement, indicating shared priors rather than shared knowledge.
Provability is therefore not a parameter of the forensic environment. It is state-dependent.

\paragraph{The framing.} An earlier version of this argument \citep{bauer2026a} concluded that
the market dies of an information crisis. That conclusion was not established, because a liability commitment faces
three ceilings, not one, and the earlier argument ordered only two of them. Separation requires
\begin{equation}
\underbrace{\frac{v}{\theta_{\mathrm{eff}}}}_{\text{must be posted}}
\;\le\;
\underbrace{\min\{\bar L,\; mv\}}_{\text{can be posted}},
\label{eq:three}
\end{equation}
where $\bar L$ is the provider's credible capacity and $m$ the multiple of actual loss that
survives review under the penalty doctrine. The provability requirement is a lower bound on the
commitment; solvency and doctrine are upper bounds. Which upper bound binds is a question about
the ratio $\ell = \bar L/v$ relative to $m$, and it has an economic content the earlier framing
suppressed: \emph{small engagements are constrained by law, large engagements by capital}, with
the switch at $v^{*} = \bar L/m$. This recovers, in a different guise, the judgment-proof logic
of \citet{shavell1986} and \citet{summers1983}, in which the effective liability ceiling is the
minimum of the legal award and the reachable assets.

\paragraph{Results.} Writing $M = \min\{\ell, m\}$ and $\bar\kappa = \sup_g \kappa(g) \le 1$, separation survives at
every capability level if and only if $\xi\bar\kappa \le 1 - 1/(M\theta_0)$; we report the
conservative benchmark $\bar\kappa = 1$. The threshold is steeply increasing in $M$ over
$[1,2]$ and flat above $M=3$: at $M = 1.2$ and $\theta_0 = 0.85$ it is $0.020$; at $M=2$,
$0.412$; at $M=5$, $0.765$. This replaces the point estimate of the earlier version, which fixed
$m = 1.2$ without justification --- a value we now regard as indefensible, since post-\emph{Cavendish}
English law fixes no multiple at all.

The regime map does the work. Calibrated to six engagement types across four jurisdictions, three
of twenty-four cells are solvency-bound when professional indemnity cover responds. After the
generative-AI exclusions introduced from January 2026 --- ISO/Verisk endorsements CG 40 47,
CG 40 48 and CG 35 08, and absolute exclusions such as W.R.\ Berkley form PC 51380 --- eleven of
twenty-four are solvency-bound, and the cells surviving $\xi = 0.4$ fall from twelve to four.
Common-law jurisdictions are doctrine-bound and were already at the boundary; civil-law B2B
contracting, where \S\,348 HGB withholds judicial reduction from merchants, is solvency-bound and
is where the insurance withdrawal bites.

\paragraph{What we do not claim.} We do not claim that the informational value of the liability
commitment goes to zero: $\theta_{\mathrm{eff}}D$ remains strictly positive for $\xi<1$, and
\citet{bauer2026b} is correct that the signal attenuates rather than vanishes. We do not claim
that one channel dominates. We do not claim the parameters are estimated. And we restrict the
scope: where verification is deterministic --- compilation, formal proof, reconciliation --- the
sharing parameter $\xi$ is structurally low and the mechanism does not apply. Section~\ref{sec:scope}
states these limits and what would falsify the argument.

\section{Related literature}\label{sec:lit}

\paragraph{Signalling and warranties.} The apparatus is \citet{spence1973} and \citet{riley1979},
with selection following \citet{cho1987} and existence with a continuum of types following
\citet{mailath1987}. That a warranty or liability commitment can signal unobserved quality is
established \citep{grossman1981,gal-or1989,lutz1989}. \citet{spence1977} shows full producer
liability is optimal under perfect competition; \citet{polinsky1983} qualify this for imperfect
competition. \citet{daughety1995} is the closest antecedent, with \citet{daughety2008} extending
to a unified model of disclosure and signalling --- a paper that also supplies the key negative
result for the disclosure alternative discussed in Section~\ref{sec:scope}.

\paragraph{Credence goods.} \citet{darby1973} name the category; \citet{dulleck2006} provide the
unifying framework in which \emph{liability} and \emph{verifiability} are the two institutional
assumptions that matter here; \citet{dulleck2011} establish experimentally that these are the
decisive institutions. \citet{erlei2025} studies generative AI in credence-goods markets
experimentally and finds that human--AI--human markets outperform human--human markets under
transparency rules but that the gains vanish under obfuscation.

\paragraph{Probabilistic enforcement.} That a stake must be grossed up by the inverse of the
probability of prevailing is standard in the enforcement literature \citep{becker1968,polinsky2000}
and in the litigation and settlement literature \citep{bebchuk1984}. \citet{craswell1999} supplies
the qualification that matters here: the simple multiplier is valid only when the probability of
sanction is independent of the state. \citet{spier1992} is the closest antecedent for a signalling
model in which ex-post enforcement and verification costs generate the signalling friction, though
without a provability gross-up or a ceiling. On limited liability in contracting see
\citet{sappington1983} and \citet{innes1990}.

\paragraph{Judgment-proof providers.} \citet{shavell1986} and \citet{summers1983} establish that
where assets fall short of harm, the effective liability ceiling is the reachable capital rather
than the legal award. \citet{che2008} show that entrepreneurs may structure finance strategically
to shield assets, and \citet{boyd1994} study damage caps under potential insolvency --- the legal
ceiling in isolation. This is the literature that our solvency ceiling instantiates. The
construction $M=\min\{\ell,m\}$ joins the asset side and the cap side, and the ticket-size switch
$v^{*}=\bar L/m$ that results appears, to our knowledge, to be new.

\paragraph{Algorithmic monoculture.} \citet{kleinberg2021} show that monocultural convergence on a
common algorithm can reduce collective decision quality. \citet{gorecki2025}, evaluating fifty
language models across six prediction tasks, find that the empirical landscape sits between
monoculture and multiplicity: systematic exclusion without recourse is rare, but model similarity
is real. We read this as support for our Assumption~\ref{ass:shift} together with a caution
against setting $\xi = 1$. \citet{kim2025} and \citet{goel2025} supply the quantitative core,
and \citet{kuai2026} close the apparent escape route: auditing behavioural entanglement across
eighteen models from six families, they find that shared pretraining data, distillation and
alignment induce latent dependencies that undermine ensemble verification precisely because such
schemes assume independent signals, so assembling nominally diverse verifiers does not restore
independence. None of this literature considers contractual enforcement.

\paragraph{Deskilling.} \citet{budzyn2025} provide the first real-world clinical evidence:
across four Polish centres, the adenoma detection rate in unassisted colonoscopy fell from
28.4\% (226/795) before routine AI exposure to 22.4\% (145/648) afterwards, a 6.0 percentage-point
absolute decline, with an adjusted odds ratio of 0.69 (95\% CI 0.53--0.89). The study is
observational, not randomised, and we treat it as such. \citet{humlum2025} estimate precise null
effects of generative AI on Danish earnings and hours, ruling out effects above 2\%. We treat
this as a discipline on the form of the claim: our thresholds are capability thresholds and
regime boundaries, not calendar predictions.

\section{Model}\label{sec:model}

\subsection{Environment}

Providers are indexed by fallback capability $s \in [s_F, \alpha)$, where $\alpha$ is AI
capability. Each serves engagements of value $v$. The AI fails with probability $\pi$;
conditional on failure the provider rescues the engagement with probability
$\rho(s) = \rho_{\max}(s/\alpha)^{a}$, $a \in (0,1)$, so $\rho' > 0$ and $\rho'' < 0$.

Capability is endogenous, built by routing a share $h$ of engagements through human review and
depreciating through disuse:
\begin{equation}
s' = s + (\alpha-s)(1-\pi)\big[\varphi\pi h - \gamma(1-h)\big].
\label{eq:dyn}
\end{equation}

\begin{lemma}[Bang-bang structure]\label{lem:bang}
With $k(h) \equiv (1-\pi)[\varphi\pi h - \gamma(1-h)]$, \eqref{eq:dyn} gives
$s'-\alpha = (s-\alpha)(1-k)$, so the steady state is $\alpha$ if $k>0$ and $s_F$ otherwise, with
threshold $h_{\min} = \gamma/(\varphi\pi+\gamma)$, strictly decreasing in $\pi$. Only two
policies can be optimal: build at $h_{\min}+\epsilon$, or abandon.
\end{lemma}

That $h_{\min}$ falls in $\pi$ is the reliability paradox: as the AI improves, a \emph{larger}
share of engagements must pass through humans merely to hold capability constant, because
failures are the only instructive events.

\subsection{The liability signal and its three ceilings}

A provider may post an outcome-contingent commitment $L$, payable when the engagement fails and
the failure is established. With $\theta$ the probability that failure is established ex post,
\begin{equation}
C(s,L) = \theta\pi(1-\rho(s))L,
\qquad
\frac{\partial^{2}C}{\partial s\,\partial L} = -\theta\pi\rho'(s) < 0,
\label{eq:cost}
\end{equation}
so single crossing holds and the least-cost separating schedule is
$L^{*}(s) = c_0 + (\chi/\theta\pi)\ln[(1-\rho(s_F))/(1-\rho(s))]$.

Three constraints bound the commitment: one from below and two from above.

\begin{lemma}[The provability floor]\label{lem:floor}
Let the commitment make the client whole in the states where the engagement fails and is not
rescued, and let $\theta_{\mathrm{eff}} \equiv \Prob(\text{established} \mid \text{failure})$ be
the probability that a failure is established ex post. Full compensation requires
\[
\theta_{\mathrm{eff}}\,\pi\,\big(1-\rho(s)\big)\,L \;\ge\; \pi\,\big(1-\rho(s)\big)\,v ,
\]
in which $\pi$ and $1-\rho(s)$ cancel, leaving
\begin{equation}
L \;\ge\; \frac{v}{\theta_{\mathrm{eff}}} .
\label{eq:floor}
\end{equation}
\end{lemma}

The cancellation is what makes \eqref{eq:floor} a clean floor: it is independent of the failure
rate and of the type, because both scale the loss and the payout identically.

\begin{remark}[Full compensation is a premise, and where it holds]\label{rem:full}
In warranty-signalling models with consumer-side moral hazard, partial coverage is the typical
equilibrium: \citet{cooper1985} obtain incomplete insurance under double moral hazard,
\citet{lutz1989} signals high quality with a \emph{low} warranty, and in \citet{matthews1987}
coverage is generically partial and non-monotone in type. Here the client takes no action that
affects the failure probability --- the engagement is produced by the provider's AI and rescued,
or not, by the provider's capability --- so the moral-hazard rationale for partial coverage is
absent and full coverage is the natural benchmark, as in \citet{grossman1981}. Equivalently,
\eqref{eq:floor} is the client's participation condition: below the floor the expected indemnity
$\theta_{\mathrm{eff}}L$ falls short of the loss $v$, and a risk-neutral client declines any
price that embeds a competence premium. If client-side moral hazard were present, partial
coverage would re-emerge and the floor would overstate the required commitment; the direction of
that bias matches the endogenous-$\theta$ case in Section~\ref{sec:scope}.
\end{remark}

\begin{remark}[What the floor is not]\label{rem:notfloor}
Two readings must be set aside. First, \eqref{eq:floor} is \emph{not} a no-mimicry condition.
Requiring that the fringe type not profit from posting a higher type's commitment yields the
differential equation $\chi\rho'(s) = \theta\pi(1-\rho(s))L^{*\prime}(s)$ whose solution is the
schedule itself; the no-mimicry requirement \emph{is} $L^{*}(s)$, not an additional floor beneath
it. Second, \eqref{eq:floor} coincides algebraically with the enforcement multiplier of
\citet{becker1968} and \citet{polinsky2000}, under which the optimal sanction equals harm divided
by the probability of detection. The coincidence is formal only. That multiplier presupposes an
offence that is chosen and can be deterred; here the AI failure is stochastic and the provider's
decision margin is the fallback investment $h$ of Lemma~\ref{lem:bang}, not the failure itself.
We therefore rest \eqref{eq:floor} on compensation rather than deterrence.
\end{remark}

\begin{remark}[Why the floor is stated in $\theta_{\mathrm{eff}}$]\label{rem:cond}
\eqref{eq:floor} conditions on failure. This is not a refinement but a requirement.
\citet{craswell1999} shows that once the probability of sanction depends on the state, the
multiplier computed case by case, on average, or as a constant diverge, so that equating the
expected sanction to harm divided by an \emph{unconditional} probability does not generally
achieve the intended standard. The warning binds here rather than incidentally: the mechanism of
Section~\ref{sec:collapse} is precisely that the verifier fails in the states where the error is
common, so failure and non-establishment are positively correlated by construction. Writing the
floor with an unconditional $\theta_0$ would understate it by the factor $(1-\xi\kappa)$.
\end{remark}

Solvency caps credible capacity at $\bar L$, and the penalty doctrine caps the enforceable term at
$mv$: the two upper bounds. Separation requires \eqref{eq:three}. Dividing by $v$ and writing
$\ell \equiv \bar L/v$ for the solvency ratio:
\begin{equation}
\frac{1}{\theta} \;\le\; M, \qquad M \equiv \min\{\ell,\, m\}.
\label{eq:compact}
\end{equation}

\begin{lemma}[The floor and the schedule]\label{lem:reconcile}
The schedule $L^{*}(s)$ and the floor \eqref{eq:floor} are logically independent, so the floor is
not redundant. The implemented commitment is $L(s) = \max\{L^{*}(s),\, v/\theta_{\mathrm{eff}}\}$.
Since $L^{*}(s_F) = c_0$ and $L^{*}$ is strictly increasing, the floor binds on a lower interval
$[s_F, \tilde s]$, where $\tilde s$ solves $L^{*}(\tilde s) = v/\theta_{\mathrm{eff}}$ whenever
such a $\tilde s$ exists in the type space. On that interval the schedule cannot be implemented,
types pool at the floor and separation fails --- bunching in the screening sense of
\citet{mussa1978} and \citet{maskin1984}, transplanted here to a signalling model in which the
pooling is induced by an exogenous enforcement floor rather than by the ironing of a menu; above
$\tilde s$ the schedule governs and the floor is slack.
\end{lemma}

Lemma~\ref{lem:reconcile} has an interpretation. Because $L^{*}$ is anchored at the enforcement
cost floor $c_0$ while the provability floor scales with $v$, small engagements are the ones on
which the two collide: below a minimum ticket size no credible commitment separates any types at
all, which is the mirror image of the maximum ticket size derived in
Proposition~\ref{prop:regime}. One reading must be blocked here: in the screening literature the
pooling of types can be the principal's rent-minimising optimum, not a failure. That logic does
not transfer, for two reasons. There is no menu-designing principal in this market --- pooling is
forced by an enforcement floor, not chosen --- and the information that pooling destroys is not a
static allocation margin but the dynamic one: it is what sustains the provider's incentive to
keep building fallback skill at all \citep{bauer2026b}. Pooling here is deskilling with a lag,
which is why we count it as a collapse and not as an optimum.

\begin{remark}[Why the ordering is not a technicality]\label{rem:order}
$M$ is a minimum, so the comparative statics of the model switch with the identity of the binding
ceiling. A change in $\theta$ propagates to market outcomes through $\ell$ in one region of the
parameter space and through $m$ in another, and an intervention that relaxes the slack ceiling has
no effect at all. The earlier version of this argument compared the verifiability channel against
the pure attenuation channel and never asked whether $\ell$ or $m$ binds --- which is why its
conclusion about \emph{why} the market fails was unsupported even where its arithmetic was right.
\end{remark}

\subsection{Enforceability across jurisdictions}

The multiple $m$ is a legal datum and varies. In England, \emph{Cavendish Square Holding BV v
Talal El Makdessi} and \emph{ParkingEye Ltd v Beavis} [2015] UKSC 67 replaced the
genuine-pre-estimate test with a proportionality test: a secondary obligation is penal if it
imposes a detriment ``out of all proportion to any legitimate interest of the innocent party in
the enforcement of the primary obligation''. Two features matter. First, the rule reaches only
\emph{secondary} obligations triggered by breach; primary obligations and price-adjustment terms
escape review entirely, so a well-drafted commitment may face no ceiling at all. Second,
\emph{Cavendish} states a presumption of enforceability between sophisticated parties with legal
advice. There is accordingly no defensible point value for $m$ in England, and we report a range. The
verified post-2015 span is modest: in \emph{Houssein v London Credit Ltd} the Court of Appeal
corrected the test ([2024] EWCA Civ 721) and, on remission, a default rate at four times the
standard rate was held non-penal ([2025] EWHC 2749 (Ch)), a result the Court of Appeal has since
affirmed on a second appeal ([2026] EWCA Civ 830); we accordingly centre the English
multiple at three with a range of 1.2--4.
In the United States, Restatement (Second) of Contracts \S\,356 and UCC \S\,2-718 retain a
loss-referenced reasonableness test, giving $m \approx 1$. In Germany, \S\,343 BGB permits
judicial reduction but \S\,348 HGB withholds it from merchants acting in the course of business,
leaving only the residual limits of \S\,242 BGB and, for standard terms, \S\,307 BGB. The
distinction matters: in individually negotiated terms $m$ is large though not unbounded, while in
standard terms it is small --- the BGH holds a construction-contract penalty clause in general
terms invalid where its cap exceeds five per cent of the contract sum (BGH, 15 February 2024,
VII~ZR~42/22, continuing BGHZ 153, 311), so standard-form $m$ sits near one. Our German range
applies to negotiated B2B contracts only. In Austria, \S\,1336(2) ABGB mandates moderation
including, on the prevailing view, between undertakings.\footnote{The prevailing Austrian view
holds the moderation right mandatory and not waivable in advance, protecting undertakings as
well; a minority position in the practitioner literature treats it as waivable in B2B (Zak
2021/7). We follow the prevailing view and mark the value as contested in
Appendix~\ref{app:prov}.} Throughout we treat $m$ as exogenous; Section~\ref{sec:scope} sets out
the consequence, which cuts against our own result.

\section{The verification collapse}\label{sec:collapse}

\subsection{Decomposing the residual}

\begin{definition}[Common and idiosyncratic error]\label{def:err}
Let $g \ge 0$ index AI capability. Error decomposes as
$\pi(g) = \pi_c(g) + (1-\pi_c(g))\pi_\iota(g)$ with common share
$\kappa(g) \equiv \pi_c(g)/\pi(g) \in (0,1)$. Write $\bar\kappa \equiv \sup_g \kappa(g) \le 1$.
\end{definition}

\begin{assumption}[Compositional shift]\label{ass:shift}
$\kappa'(g) > 0$: capability growth eliminates idiosyncratic error faster than common error.
\end{assumption}

Assumption~\ref{ass:shift} is empirical, not a convenience. \citet{kim2025} report that larger
and more accurate models have highly correlated errors even with distinct architectures and
providers; \citet{goel2025} report that model mistakes become more similar as capabilities
increase. \citet{kuai2026} audit the mechanism directly: shared pretraining corpora, distillation
and alignment pipelines induce latent behavioural entanglement across model families, and
stronger dependence between judge and generator is associated with over-endorsement by the judge.
\citet{pombal2026} sharpen the judge-side result: self-preference bias persists even under fully
objective, rubric-based evaluation, and ensembling multiple judges mitigates it without
eliminating it.
\citet{gorecki2025} qualify the picture --- multiplicity is real and systematic
exclusion without recourse is rare --- which is why we treat $\xi$ as a parameter to be varied
rather than set to one, and why we carry $\bar\kappa$ as a free parameter rather than forcing
$\bar\kappa = 1$; the mid-range values $\xi \in [0.2, 0.4]$ that organise the calibration sit
deliberately between those poles.

\begin{remark}[Identification]\label{rem:ident}
To separate composition from level, we hold the total error path exogenous and vary only
composition: $\pi(g) = \pi_0e^{-bg}$ and $\kappa(g) = \bar\kappa-(\bar\kappa-\kappa_0)e^{-\delta g}$, with
$\bar\kappa \in (\kappa_0, 1]$. The null $\delta=0$ is the falsification benchmark, and no
statement about $\delta$ depends on the speed $b$ of capability growth. The calibration sets
$\bar\kappa = 1$ as the conservative benchmark; every threshold below scales as $1/\bar\kappa$.
\end{remark}

\subsection{State-dependent provability}

\begin{proposition}[Provability under model sharing]\label{prop:theta}
Let $\xi \in [0,1]$ be the intensity with which verifier and generator draw on the same
foundation model, and let $\nu \in [0,1]$ index gains in \emph{base} provability from better
forensic technology. Then
\begin{equation}
\theta_{\mathrm{eff}}(g) = \theta_0(g)\big(1-\xi\kappa(g)\big),
\qquad
\theta_0(g) = \theta_0 + \nu(1-\theta_0)\big(1-e^{-b_vg}\big),
\label{eq:thetaeff}
\end{equation}
strictly decreasing in $\xi$, in $\kappa$, and hence in $g$.
\end{proposition}

The content is that enforcement requires evidence \emph{in the failure state}. Conditional on a
failure the event is common with probability $\kappa$, and in that state the verifier shares the
blind spot with probability $\xi$. Figure~\ref{fig:mech} shows the two halves. Panel~(a) is
Definition~\ref{def:err} under Assumption~\ref{ass:shift}: total error falls while the common
share rises, so the residual on which liability is priced becomes progressively common.
Panel~(b) carries that composition into \eqref{eq:thetaeff}. Two features are worth noting. The
$\xi = 0$ line is flat, because with independent verification the compositional shift is
irrelevant --- it is only through model sharing that the shift has any bite. And the decline is
bounded below by $\theta_0(1-\xi)$, which is Corollary~\ref{cor:nu} read off the figure.

\begin{figure}[t]\centering
\includegraphics[width=0.86\textwidth]{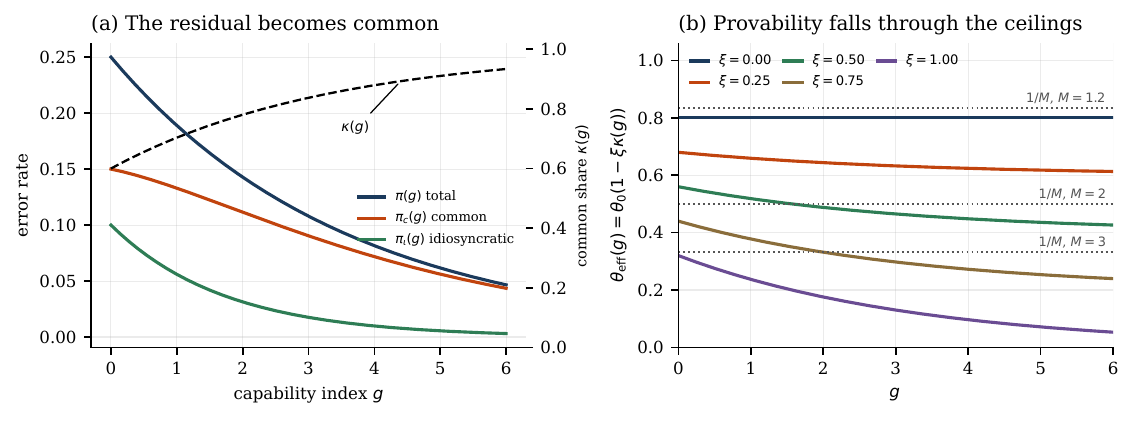}
\caption{The mechanism, at $\theta_0 = 0.80$, $\kappa_0 = 0.60$, $\delta = 0.30$, $b = 0.28$.
(a)~The residual becomes common. (b)~Effective provability under model sharing, against the
thresholds $1/M$ implied by three values of the binding ceiling. A curve lying below a dotted line
is a configuration in which separation is infeasible at that ceiling. At $\theta_0 = 0.80$ the
ceiling $M = 1.2$ fails even at $\xi = 0$, since $M\theta_0 < 1$.}\label{fig:mech}
\end{figure}

\subsection{The survival condition}

\begin{proposition}[Survival]\label{prop:survival}
Separation is feasible if and only if $\theta_{\mathrm{eff}} \ge 1/M$, that is
\[
\xi\kappa \;\le\; 1-\frac{1}{M\theta_0},
\qquad M = \min\{\ell,\,m\}.
\]
Separation survives at \emph{every} capability level, i.e.\ as $\kappa \to \bar\kappa$, if and
only if
\begin{equation}
\boxed{\;\xi\,\bar\kappa \;\le\; 1-\frac{1}{M\theta_0}\;}
\label{eq:survival}
\end{equation}
If $M\theta_0 \le 1$ separation is infeasible even at $\xi = 0$.
\end{proposition}

Table~\ref{tab:threshold} evaluates \eqref{eq:survival} at $\bar\kappa = 1$; if multiplicity
holds $\bar\kappa$ below one \citep{gorecki2025}, each entry scales up by $1/\bar\kappa$. The
threshold is steeply increasing in
$M$ over $[1,2]$ and flattens above $M = 3$. This region matters because it is where the
institutions sit: US doctrine puts $M \approx 1$, Austrian moderation $M \approx 1.5$, and --- as
Section~\ref{sec:regimes} shows --- the post-exclusion solvency ratio puts $M$ between $0.3$ and
$1.5$ for mid-size and large engagements.

\begin{table}[t]\centering\small
\caption{Critical model sharing $\xi^{*} = 1-1/(M\theta_0)$, evaluated at $\bar\kappa = 1$;
entries scale as $1/\bar\kappa$. ``n.f.'' denotes infeasible even at $\xi = 0$.}\label{tab:threshold}
\begin{tabular}{@{}lrrrrr@{}}
\toprule
$M = \min\{\ell,m\}$ & $\theta_0 = 0.70$ & $0.78$ & $0.82$ & $0.85$ & $0.95$\\
\midrule
1.0  & n.f.  & n.f.  & n.f.  & n.f.  & n.f.\\
1.2  & n.f.  & n.f.  & n.f.  & 0.020 & 0.123\\
1.5  & 0.048 & 0.145 & 0.187 & 0.216 & 0.298\\
2.0  & 0.286 & 0.359 & 0.390 & 0.412 & 0.474\\
3.0  & 0.524 & 0.573 & 0.593 & 0.608 & 0.649\\
5.0  & 0.714 & 0.744 & 0.756 & 0.765 & 0.789\\
10.0 & 0.857 & 0.872 & 0.878 & 0.882 & 0.895\\
\bottomrule
\end{tabular}
\end{table}

\begin{corollary}[Verifier quality is the wrong margin]\label{cor:nu}
Since $\theta_{\mathrm{eff}} \to \theta_0(1-\xi\bar\kappa)$ as $\kappa \to \bar\kappa$, even
$\theta_0 = 1$ leaves $\theta_{\mathrm{eff}} \ge 1-\xi\bar\kappa$ and hence requires
$M \ge 1/(1-\xi\bar\kappa)$. At $\bar\kappa = 1$, $\xi = 0.5$ needs a ceiling of 2 and
$\xi = 0.75$ a ceiling of 4, whatever the quality of the forensic apparatus.
\end{corollary}

Corollary~\ref{cor:nu} follows because $\nu$ multiplies $\theta_0$ and leaves the bracket
in \eqref{eq:thetaeff} unchanged: a verifier that is down is down regardless of its quality. It
inverts the natural policy recommendation, and \citet{denisov2026} supply the independent
technical statement of the same point.

\subsection{The demand side: what the client still pays for}\label{sec:demand}

The stakes $\chi$ riding on the signal have so far been a primitive, inherited from
\citet{bauer2026b}. They can be microfounded from the same client stage, and doing so
reveals that $\xi\kappa$ acts on the signal twice --- once on its cost, once on its demand
--- with a sign that depends on which ceiling binds. A client with ticket $v$ values a
skilled provider for the losses it avoids, \emph{net of what compensation would have
replaced anyway}: with expected net recovery $R = \theta_{\mathrm{eff}}L - c_0$ per
unrescued failure, the stakes on the skill distinction become
\begin{equation}
\chi_{\mathrm{eff}} \;=\; \chi\Bigl(1 - \frac{\theta_{\mathrm{eff}}L}{v} + \frac{c_0}{v}\Bigr).
\label{eq:chieff}
\end{equation}
Two regimes follow. Where the compensation floor binds, $L = v/\theta_{\mathrm{eff}}$,
the bracket collapses to $c_0/v$: full compensation crowds out the premium for
competence --- the warranty logic of \citet{spence1977} --- and only the enforcement
wedge $c_0$ keeps the client caring who does the work. Where a ceiling binds instead,
$L = \bar L < v/\theta_{\mathrm{eff}}$, the bracket \emph{rises} as
$\theta_{\mathrm{eff}}$ falls: in the calibration ($v = 1$, $c_0 = 0.1$, $\bar L = 1$)
it moves from $0.30$ at $\theta_{\mathrm{eff}} = 0.80$ to $0.86$ at
$\theta_{\mathrm{eff}} = 0.24$.

\begin{corollary}[Two margins of collapse]\label{cor:demand}
Under a binding solvency or penalty ceiling, model sharing destroys the signal from the
cost side --- the separating pledge scales as $1/\theta_{\mathrm{eff}}$ --- while the
demand for exactly this certification grows. The market that collapses is one whose
surplus is increasing as it collapses, so a cost-side account alone understates the
welfare loss. At the compensation floor the signal is instead doubly dead: too expensive
to send and nearly worthless to receive.
\end{corollary}

The alternative reading --- that unprovability makes \emph{all} promises of the
relationship cheap talk, so the extractable share $\zeta$ and hence
$\chi_{\mathrm{eff}}$ \emph{fall} with $\theta_{\mathrm{eff}}$ --- is institutional
rather than derivable from the client stage, and the two are empirically separable:
they predict opposite movements of competence premia in segments where verifiability
deteriorates. We leave the discrimination to data; either way, nothing in
Sections~\ref{sec:regimes}--\ref{sec:scope} depends on $\chi_{\mathrm{eff}}$, which
prices the signal's value, not its feasibility.

\section{Which ceiling binds}\label{sec:regimes}

\begin{proposition}[Regime partition]\label{prop:regime}
The parameter space partitions at $\ell = m$. For $\ell > m$ the market is
\emph{doctrine-bound}: $M = m$, and $\xi^{*}$ is independent of the provider's capital. For
$\ell < m$ it is \emph{solvency-bound}: $M = \ell = \bar L/v$, and $\xi^{*}$ falls in the
engagement value. The switch occurs at the ticket size
\begin{equation}
v^{*} = \bar L/m .
\label{eq:vstar}
\end{equation}
Consequently $v_{\max} = \bar L\,\theta_{\mathrm{eff}}$ is the largest engagement that can be
signalled at all, provided the doctrine condition $m\theta_{\mathrm{eff}} \ge 1$ holds; if it
fails, no ticket size works.
\end{proposition}

Proposition~\ref{prop:regime} has an immediate reading. Small engagements are constrained by law:
the provider has ample capacity relative to the ticket, and the question is what a court will
enforce. Large engagements are constrained by capital: the law would enforce more than the
provider can credibly promise. The signalling literature has treated the ceiling as a single
object; the judgment-proof literature \citep{shavell1986,summers1983,che2008} has treated the
capital constraint in isolation from information. The regime map is where they meet.

\begin{figure}[t]\centering
\includegraphics[width=\textwidth]{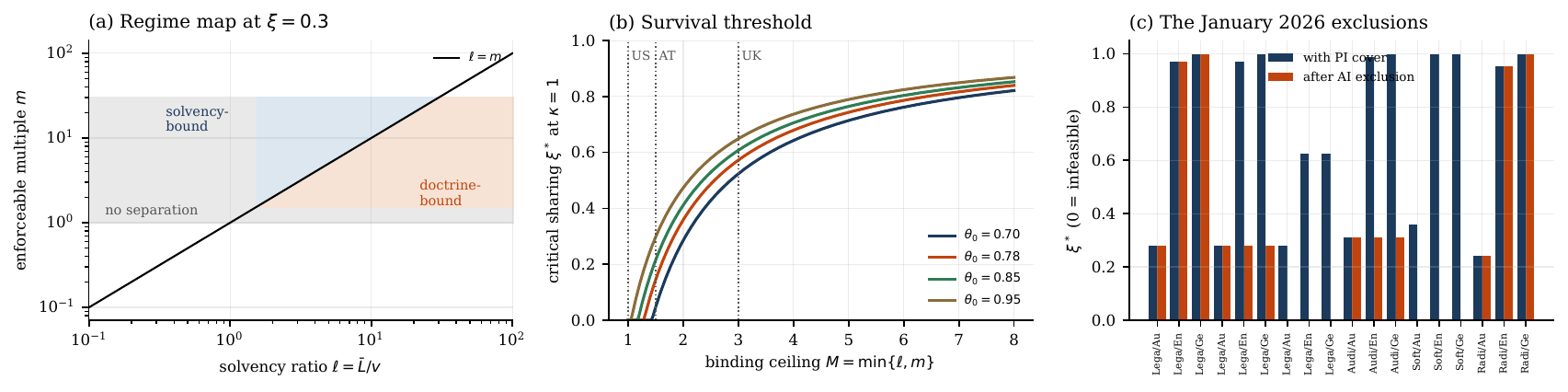}
\caption{(a) The regime partition in $(\ell, m)$ space at $\xi = 0.3$, $\theta_0 = 0.80$,
$\kappa = 0.60$; the diagonal is $\ell = m$. (b) The survival threshold as a function of the
binding ceiling, with the jurisdictional values marked. (c) Critical model sharing by
engagement--jurisdiction cell, with and without responding professional indemnity cover; zero
denotes infeasibility at any $\xi$.}\label{fig:regimes}
\end{figure}

\subsection{The January 2026 exclusions}

Credible capacity in professional services is supplied predominantly by the professional
indemnity policy limit rather than the balance sheet, and that limit has come under attack.
ISO/Verisk endorsements CG 40 47, CG 40 48 and CG 35 08 (edition 01 26) exclude
generative-AI-related liability from commercial general liability cover and became available for
attachment from 1 January 2026; W.R.\ Berkley form PC 51380 (06-24) is an absolute AI exclusion
in professional and management lines. We model this as a fall in $\bar L$ from the insured level
to own capital, set at 15\% of the insured figure. The share is an assumption; the robustness
band is reported below Corollary~\ref{cor:shift}. The exclusion forms themselves do not state
correlated model error as the ground; the closest wording-level framing is the silent-cyber
analogy --- unpriced aggregation risk inside legacy wordings. Market analysis has, however, begun
to name the mechanism: \citet{leung2026}, mapping the emerging insurability frontier of AI risk,
identify foundation-model concentration as its most genuinely novel feature, because the failure
of a shared upstream model produces correlated losses across the book rather than independent
ones. This is consistent with, though it does not by itself confirm, the mechanism of
Section~\ref{sec:collapse}.

\begin{table}[h]\centering\small
\caption{Regime by engagement and jurisdiction. ``D'' denotes doctrine-bound, ``S''
solvency-bound. Left block: professional indemnity cover responds. Right block: after
generative-AI exclusion.}\label{tab:regimemap}
\begin{tabular}{@{}lrrcccccccc@{}}
\toprule
& & & \multicolumn{4}{c}{with PI cover} & \multicolumn{4}{c}{after exclusion}\\
\cmidrule(lr){4-7}\cmidrule(lr){8-11}
Engagement & $\bar L$ (m) & $v$ (m) & US & AT & UK & DE & US & AT & UK & DE\\
\midrule
Legal, small matter   & 3  & 0.025 & D & D & D & D & D & D & D & D\\
Legal, mid matter     & 5  & 0.5   & D & D & D & D & D & D & S & S\\
Legal, large matter   & 20 & 10    & D & D & S & S & S & S & S & S\\
Audit, mid client     & 10 & 1     & D & D & D & D & D & D & S & S\\
Software, mid project & 2  & 0.3   & D & D & D & S & D & S & S & S\\
Radiology, per study  & 5  & 0.05  & D & D & D & D & D & D & D & D\\
\bottomrule
\end{tabular}
\end{table}

\begin{corollary}[The exclusions move the regime boundary]\label{cor:shift}
In the calibration of Table~\ref{tab:regimemap}, three of twenty-four cells are solvency-bound
when cover responds and eleven after exclusion. The cells that survive $\xi = 0.2$ fall from
eighteen to twelve, and those surviving $\xi = 0.4$ from twelve to four.
\end{corollary}

Substituting the occupation-level verifiabilities of \citet{bauer2026c} for the
engagement-type values --- the wrong object, as Appendix~B explains, but the natural
robustness check --- moves the counts to twelve and ten (insured, $\xi = 0.2$ and $0.4$) and
seven and four (post-exclusion): every flip but one is a legal cell, through that paper's
causation discount, and the headline post-exclusion count of four cells at $\xi = 0.4$ is
invariant. The own-capital share behind these counts is an assumption, so we report the band. Over shares
of 5--25 per cent, the post-exclusion solvency-bound count stays between ten and eighteen of
twenty-four, against three with cover, so the regime shift is not an artefact of the 15 per cent
figure. The count of cells surviving $\xi = 0.4$ is more sensitive: four at shares up to 15 per
cent, ten at 25 per cent. The qualitative statement --- the exclusions move the binding
constraint --- is robust; the severity is not.

Two features of Corollary~\ref{cor:shift} are worth isolating. First, the common-law cells are
largely unmoved, because they were already doctrine-bound and, at $m \approx 1$, already
infeasible: $M\theta_0 < 1$ regardless of capital. Second, the cells that move are precisely the
civil-law B2B cells, where \S\,348 HGB leaves $m$ large and the binding ceiling was therefore the
capital constraint all along. The insurance withdrawal is a civil-law event.

This is the sense in which the framing correction matters. An account that attributes market
failure to the collapse of provability is right about the mechanism in the doctrine-bound region
and incomplete in the solvency-bound region, where a provider with adequate provability may still
be unable to post a credible commitment. Both channels are real; the regime map says where each
applies.

\section{Scope, robustness and falsification}\label{sec:scope}

\paragraph{Where the mechanism does not apply.} Where verification is deterministic --- compilation,
type checking, formal proof, reconciliation against an independent ledger --- the verifier does
not draw on the generator's model and $\xi$ is structurally near zero. The mechanism therefore
does not apply to that part of the verification problem. It applies to the part that is not
deterministic, which is larger than it appears: \citet{yang2024} find that for 87.13\% of defects
LLMs generate no valid unit test at all, and that among the remainder only 47.28\% of defects are
detected (their Finding~8 and Table~6). Compilation catches syntax and types, not semantics. In law, medicine and audit the
review layer is human or model-based throughout. We restrict the claim to credence-goods services
in which ex-post review is judgmental rather than mechanical.

\paragraph{The empirical anchor.} $\kappa_0 = 0.60$ is taken from \citet{kim2025}'s
``agreement rate when both models err'' --- the mean across model pairs on the Helm leaderboard,
against a random-choice baseline near one third; the companion value on HuggingFace is $0.423$.
The symbol is this model's common-share parameter, not a chance-corrected kappa coefficient, and
the statistic is an order of magnitude for that share, not an estimate of $\pi_c/\pi$: it is
benchmark-specific, which is one reason the sensitivity band below is wide. Over
$\kappa_0 \in [0.2, 0.8]$ the critical capability falls monotonically by roughly half in our
calibration. Direction is robust; level is not.

\paragraph{The enforceable multiple.} We report $m$ as a range, not a point. An earlier version
fixed $m = 1.2$ for England, inherited from \citet{bauer2026c} without re-derivation and used one
sentence after the text itself observed that no fixed multiple is defensible. That value should
not be used; the second version of \citet{bauer2026c} corrects it to the same central multiple
and range adopted here. It matters: at $\theta_0 = 0.85$ the threshold moves from $0.020$ at $M=1.2$ to
$0.412$ at $M=2$. The secondary-obligation carve-out in \emph{Cavendish} makes even the range
uncertain, since a commitment drafted as a primary obligation escapes the doctrine.

\paragraph{Insurance adoption.} Adoption is now partially documented: industry filing reports
indicate that by spring 2026 major carriers --- among them W.R.\ Berkley, Chubb, Travelers,
Berkshire Hathaway and Cincinnati Financial --- had filed to adopt the ISO endorsements or
proprietary AI exclusion wording, with the large majority of state filings approved, and with
isolated non-adoption notices on the other side. Attachment to an individual policy remains
discretionary, so Corollary~\ref{cor:shift} should still be read as the consequence of full
attachment, an upper bound on the effect. On the other side of the ledger, a standalone
affirmative market has begun to form --- dedicated AI-liability programmes at Lloyd's and at
Munich Re, mapped by \citet{leung2026} and differentiated by peril --- but per-insured limits
remain in the single to low tens of millions of dollars, so the fourth falsification condition
below is not met.

\paragraph{The enforceable multiple is exogenous here, and probably is not.}
Post-\emph{Cavendish} the enforceable multiple depends on the legitimate interest the clause
protects, and in \emph{ParkingEye} that interest was the operation of a parking scheme rather
than compensation for loss. If a provider can establish that preserving human fallback capability
is itself a legitimate commercial interest --- and the argument is not obviously weak, since the
capability is what the client is buying --- then $m$ rises with the value of the very attribute
the commitment certifies. Condition~\eqref{eq:survival} would then be a fixed point rather than
an inequality in an exogenous parameter, and the fixed point may be self-sustaining: a larger
protected interest supports a larger enforceable multiple, which supports separation, which
preserves the interest. This cuts against our own conclusion. Where the protected interest is
large, the doctrine-bound cells of Table~\ref{tab:regimemap} would be more permissive than we
report, and the common-law markets we describe as already at the boundary might not be. We treat
$m$ as exogenous because we are not aware of any decision testing whether the maintenance of
professional capability qualifies as a legitimate interest under the \emph{Cavendish} test, and
because endogenising it changes the character of the problem rather than the calibration. It is,
in our view, the most promising extension of this framework.

\paragraph{Provability is treated as independent of the commitment.} We take
$\theta_{\mathrm{eff}}$ as independent of $L$. It probably is not: larger stakes induce more
investment in establishing the case --- verifiability made endogenous by contracting effort, as
in \citet{kvaloy2009} --- so $\theta_{\mathrm{eff}}'(L) > 0$ and \eqref{eq:floor}
becomes the fixed point $L = v/\theta_{\mathrm{eff}}(L)$. The direction of the bias is
determinate --- endogenous provability makes the floor easier to satisfy, so we overstate the
constraint --- but the magnitude is not, and the fixed point may fail to exist where
$\theta_{\mathrm{eff}}$ is bounded away from one by the sharing term, which is the interesting
case. Two further omissions run the other way and are worth naming: claimant litigation costs
raise the commitment required for full compensation above $v/\theta_{\mathrm{eff}}$
\citep{bebchuk1984}, while settlement rather than adjudication lowers the realised payment below
its face value. Appendix~\ref{app:data} records a companion package in which this fixed point is
solved numerically: the bias runs in the direction stated, and a settlement dip of sufficient
depth generates multiple fixed points.

\paragraph{Disclosure as an alternative signal.} If verifiable process disclosure --- system
cards, prompt and retrieval governance, audit certification --- could separate types, liability
would not be the last signal. \citet{erlei2025} finds experimentally that transparency raises
efficiency in human--AI--human credence-goods markets. But disclosure whose cost is
type-independent cannot separate: \citet{daughety2008} observe that the disclosure literature
assumes marginal cost independent of quality, which renders separation via signalling impossible
and leaves those types pooled, and \citet{board2009} shows that competition further undermines
full disclosure. Process disclosure is therefore a partial substitute --- high types disclose ---
not a general one. Whether it separates in professional services is open and, in our view, the
most valuable extension.

\paragraph{What would falsify the argument.} First, evidence that architecturally diverse
verifiers achieve genuine error independence, driving $\xi$ toward zero. The available evidence
points the other way \citep{kim2025,goel2025,denisov2026,kuai2026,pombal2026} --- \citet{kuai2026} in
particular find that apparent diversity across model families conceals latent entanglement
inherited from shared pretraining and alignment --- but this remains a measurement question.
Second, $\delta \le 0$, no compositional shift; then the prospective half of the mechanism fails
while the level effect survives. Third, evidence that liability disciplines credence-goods markets
even at low $\theta_{\mathrm{eff}}$. Fourth, and specific to the regime map: if a deep affirmative
market in AI liability cover emerges with capacity well above engagement values, $\ell$ ceases to
bind and the doctrine channel becomes dominant again. Current capacity --- roughly \$10m to \$25m
per risk across a handful of carriers \citep{leung2026} --- is small relative to large-matter
exposure.

\paragraph{Measurement of $\xi$.} Model-sharing intensity is not observable, because providers do
not disclose the provenance of the models embedded in their workflows. A crude anchor nonetheless
exists: enterprise foundation-model spending is heavily concentrated, with survey data for 2025
putting the three largest providers at roughly 40, 27 and 21 per cent of enterprise LLM spend
--- a Herfindahl index near 0.28, our own calculation from these three shares alone, so a
lower bound on concentration among the named leaders that ignores the residual --- so the
probability that an independently drawn reviewer shares the generator's provider is not small.\footnote{Menlo Ventures, \emph{2025: The State of
Generative AI in the Enterprise} (December 2025). The mid-year wave differs materially, so we
read these shares as volatile orders of magnitude.} We treat this as an order-of-magnitude
anchor for $\xi$, not an estimate. A first measurement route now exists on the behavioural side:
the audit framework of \citet{kuai2026} estimates entanglement between black-box models from
their joint failure patterns, so the sharing structure, while unobservable from provenance
disclosures, is no longer unmeasurable in principle. The unobservability of provenance remains
the principal empirical limitation and, as Section~\ref{sec:policy} argues, itself a regulatory
implication.

\section{Policy}\label{sec:policy}

Three implications follow, of which the first inverts the natural recommendation.

Mandatory logging and audit-trail standards raise $\theta_0$ and are desirable on other grounds,
but by Corollary~\ref{cor:nu} the effective ceiling is set by $1-\xi$ and is invariant to
verifier quality. The instrument that addresses the mechanism is a requirement on verifier
\emph{provenance}: that review and generation not share base weights, training corpora or
providers. This is verifiable in principle and currently unregulated.

Disclosure of model provenance is a precondition for both regulation and research, since $\xi$ is
otherwise unobservable and neither a provenance requirement nor a test of this paper's central
claim is feasible. The European framework illustrates the gap: Regulation (EU) 2024/1689 imposes
transparency duties from August 2026 and fines of up to EUR~35 million or seven per cent of
worldwide turnover, yet contains no requirement on the provenance of the verifier, and the
digital-omnibus amendments of Regulation (EU) 2026/1744 --- the AI-specific omnibus, distinct
from the still-pending data-protection omnibus --- which postponed the Act's high-risk
obligations, left that gap untouched. Liability policy has moved the other way: the Commission
formally withdrew the proposed AI Liability Directive in October 2025 (OJ C/2025/5423),
returning non-contractual AI liability to
national law, while the revised Product Liability Directive (EU) 2024/2853 extends strict
producer liability to software from December 2026. Jurisdictional fragmentation is therefore
widening rather than narrowing --- which is precisely the margin on which the regime map of
Section~\ref{sec:regimes} operates.

Finally, the regime map implies that interventions must be targeted. Relaxing the enforceable
multiple does nothing in a solvency-bound cell, and expanding insurance capacity does nothing in
a doctrine-bound one. Since the January 2026 exclusions move cells across the boundary, an
instrument calibrated to the previous regime may now address the slack constraint.

\section{Conclusion}

Outcome-contingent liability is a genuine substitute for collapsed production-cost signals, and it
is robust to capability growth in the sense usually claimed. We have argued that the residual on
which it is priced is not homogeneous, and that its composition shifts toward the common component
precisely as capability rises, so that what remains becomes progressively unprovable. And we have
argued that this mechanism operates against two other ceilings, one legal and one financial, whose
relative position determines whether it matters at all.

The market for expert services is not failing for one reason. It is constrained by law where
engagements are small and by capital where they are large, and the events of January 2026 have
moved that boundary. Whether the provability channel is the operative one is a question about
where in the regime map a market sits --- which is an empirical question, and one that the
unobservability of model provenance currently prevents anyone from answering.

\newpage
\appendix
\section{Proofs}\label{app:proofs}

\begin{proof}[Proof of Lemma~\ref{lem:bang}]
Subtracting $\alpha$ from \eqref{eq:dyn} gives $s'-\alpha = (s-\alpha)(1-k(h))$, so
$|s'-\alpha| < |s-\alpha|$ iff $k>0$; hence $s \to \alpha$ when $k>0$ and $s \to s_F$ otherwise.
Setting $k(h)=0$ yields $h_{\min}$, and
$\partial h_{\min}/\partial\pi = -\gamma\varphi/(\varphi\pi+\gamma)^{2} < 0$. Since the steady
state is a step function of $h$ and the flow cost of engagement is increasing in $h$, no interior
$h$ other than $h_{\min}+\epsilon$ can be optimal.
\end{proof}

\begin{proof}[Proof of Proposition~\ref{prop:theta}]
Condition on a failure: with probability $\kappa$ it is common, and in that state the verifier
fails with probability $\xi$. Hence
$\Prob(\text{provable}\mid\text{failure}) = \theta_0[(1-\kappa)+\kappa(1-\xi)] = \theta_0(1-\xi\kappa)$.
Monotonicity in $\xi$ and $\kappa$ is immediate and in $g$ follows from
Assumption~\ref{ass:shift}. The $\nu$ term multiplies $\theta_0$ and leaves the bracket unchanged.
\end{proof}

\begin{proof}[Proof of Proposition~\ref{prop:survival}]
The commitment must satisfy $\theta_{\mathrm{eff}}L \ge v$, so $L \ge v/\theta_{\mathrm{eff}}$,
while credibility and enforceability give $L \le \min\{\bar L, mv\}$. A feasible $L$ exists iff
$v/\theta_{\mathrm{eff}} \le \min\{\bar L, mv\}$, i.e.\ $1/\theta_{\mathrm{eff}} \le M$ with
$M = \min\{\ell,m\}$, $\ell = \bar L/v$. Substituting \eqref{eq:thetaeff} and rearranging gives
$\xi\kappa \le 1-1/(M\theta_0)$. Under the identification of Remark~\ref{rem:ident},
$\sup_g\kappa(g)=\bar\kappa$, so survival at every capability level is equivalent to the
condition holding at $\kappa=\bar\kappa$. If $M\theta_0 \le 1$ the right-hand side is non-positive and no $\xi \ge 0$
satisfies it.
\end{proof}

\begin{proof}[Proof of Proposition~\ref{prop:regime}]
$M = \min\{\ell, m\}$ with $\ell = \bar L/v$ strictly decreasing in $v$ and $m$ independent of it.
Hence $\ell > m \iff v < \bar L/m$, giving \eqref{eq:vstar}. On the doctrine-bound branch
$M = m$, so $\xi^{*}$ in \eqref{eq:survival} does not involve $\bar L$. On the solvency-bound
branch $M = \bar L/v$, so $\xi^{*} = 1-v/(\bar L\theta_0)$ is strictly decreasing in $v$. Setting
$\xi = 0$ and solving $1/\theta_{\mathrm{eff}} \le \bar L/v$ gives
$v \le \bar L\theta_{\mathrm{eff}}$; if $m\theta_{\mathrm{eff}} < 1$ the doctrine constraint fails
for every $v$, since it is independent of $v$.
\end{proof}

\section{Parameter provenance}\label{app:prov}

Every numerical parameter is listed with its source and whether it was re-derived here or
inherited. Inherited values not re-derived are marked, following the diagnosis in
Section~\ref{sec:scope} that an unexamined inherited parameter was the principal defect of the
earlier version.

\begin{center}\small
\begin{tabularx}{\textwidth}{@{}lXl@{}}
\toprule
Parameter & Source & Status\\
\midrule
$\alpha = 0.90$, $a = 0.5$ & \citet{bauer2026b}; $\alpha$ there from \citet{singh2026} & inherited\\
$\rho_{\max} = 0.95$, $s_F = 0.03$ & \citet{bauer2026b} & inherited; corrected in this version\\
$\varphi = 1$ & normalisation & ---\\
$\gamma$ & recovered from $h_{\min} = \gamma/(\varphi\pi+\gamma)$ & derived\\
$\theta_0$ by engagement type & calibrated here; see the note below & new; corrected in this version\\
$\kappa_0 = 0.60$ & \citet{kim2025}, agreement rate when both models err (Helm) & external, order of magnitude\\
$\delta = 0.30$, $b = 0.28$ & assumed; $\delta = 0$ is the falsification null & assumed\\
$m$ (US $=1$) & Restatement \S\,356; UCC \S\,2-718 & primary legal source\\
$m$ (England, central 3, range $1.2$--$4$) & \emph{Cavendish} [2015] UKSC 67; \emph{Houssein} [2024] EWCA Civ 721, [2025] EWHC 2749 (Ch), affirmed [2026] EWCA Civ 830 & range, no point value defensible\\
$m$ (Germany $=10$, range $3$--$25$) & \S\,348 HGB; residual \S\,242, \S\,307 BGB; in standard terms BGH VII~ZR~42/22 caps near 5\% ($m\approx1$) & negotiated B2B only\\
$m$ (Austria $=1.5$, range $1.2$--$2$) & \S\,1336(2) ABGB & range; prevailing view, contested\\
$\bar L$ insured & SRA minimum terms; published PI benchmarks & external, order of magnitude\\
$\bar L$ own capital $= 0.15\bar L$ & assumed & assumed\\
$\xi$ & not observable; varied over $[0,1]$ & free parameter\\
\bottomrule
\end{tabularx}
\end{center}

\noindent\textbf{Withdrawn.} $m = 1.2$ for England, used in an earlier draft and inherited from
\citet{bauer2026c}, is withdrawn as indefensible for the reasons in Section~\ref{sec:model};
the second version of \citet{bauer2026c} adopts the same correction.

\medskip
\noindent\textbf{Corrected.} The previous version listed $\rho_{\max} = 0.85$ and $s_F = 0.10$ as
inherited from \citet{bauer2026b}. Neither value appears there: the baseline in
\citet{bauer2026b} is $\rho_{\max} = 0.95$ and $s_F = 0.03$; the $0.85$ matches the software-row
$\rho_{\max}$ of the occupation table in \citet{bauer2026c}, and the $0.10$ has no source in any
of the three companion papers. Both are restored to the \citet{bauer2026b} baseline. No reported
number depends on these parameters: these four enter the model only symbolically, and the
survival condition, the threshold table and both figures involve only $M$, $\theta_0$, $\xi$ and
the common-share path $\kappa(g)$ generated by $(\kappa_0, \delta, b)$.

\medskip
\noindent\textbf{Corrected ($\theta_0$).} The previous version listed the engagement-type
verifiabilities as inherited from \citet{bauer2026c}. They are not: that paper calibrates
\emph{occupations} (software $0.80$, legal $0.50$, radiology $0.70$, audit $0.85$), while the
regime map of Table~\ref{tab:regimemap} prices \emph{contracted engagements} --- a scoped matter,
a statutory audit, a fixed-price project, a per-study read --- selected precisely for the
contractibility of an outcome-contingent commitment. The two objects differ most where
selection bites hardest: \citet{bauer2026c}'s legal $0.50$ prices the contested causal link from
\emph{advice} to loss across the occupation, whereas a scoped matter with an engagement letter
and a deliverable standard makes breach substantially more provable, hence $0.80$ here. The
values used ($0.78$--$0.85$) are this paper's own calibration: legal matter $0.80$
(professional-negligence breach against a defined retainer); statutory audit $0.82$
(ISA-documented standard, inspection findings); fixed-price software $0.85$ (contractual
acceptance criteria); teleradiology per-study $0.78$ (retrospective re-read against ground
truth). Their tight clustering, against the occupation-level span $0.25$--$0.85$, is the
selection effect, not a disagreement. Sensitivity under the occupation-level values is reported
with Corollary~\ref{cor:shift}. The counts of that corollary are evaluated at
$\bar\kappa = \kappa_0 = 0.60$ and are unchanged for $\bar\kappa \in [0.60, 0.72]$.

\section{Data availability}\label{app:data}

All numbers, tables and figures are produced by a self-contained Python package released with the
manuscript, containing the calibration module, the regime-map module, the figure scripts, the
serialised results and the random seeds. No proprietary or licensed data are used: all inputs are
either published parameter values cited in the text or synthetic values generated by the code.
The archive also carries, beyond the paper-backing results, a deposit-internal
verification battery of the collapse mechanism (capability thresholds per
occupation under channel configurations; \texttt{make\_battery.py}, with its
construction documented in the archive and the stakes modelled under the
enforceability variant of Section~\ref{sec:demand}), the producer of the
demand-side bracket values of that section (\texttt{make\_demand.py}), and a
self-verification gate (\texttt{verify\_deposit.py}) that re-derives every
producible result file and compares it to the shipped copy. The archive is
deposited under a CC BY 4.0 licence at
\href{https://doi.org/10.6084/m9.figshare.33188046}{doi:10.6084/m9.figshare.33188046}.

A separate companion package, deposited at
\href{https://doi.org/10.6084/m9.figshare.33190722}{doi:10.6084/m9.figshare.33190722},
holds the numerical verification of the three extensions identified in
Section~\ref{sec:scope}: the enforceable multiple as a fixed point, the provability floor as the
fixed point $L = v/\theta_{\mathrm{eff}}(L)$, and the insurability of correlated model failures.
It reports one result that bears directly on the present paper. Under a monotone effort response
$\theta_{\mathrm{eff}}'(L) > 0$ the fixed point lies strictly below the exogenous floor
$v/\theta_{\mathrm{eff}}$, which confirms the direction of bias claimed in Section~\ref{sec:scope}
and puts its size at roughly a fifth in that calibration: the constraint we report is
conservative. Its data are synthetic and generated from the model equations themselves, so the
package verifies and illustrates the derivations; it is not evidence for them, and nothing in
this paper rests on it.

\section*{Declarations}

\paragraph{Generative AI and AI-assisted technologies in the writing process.} During the
preparation of this work the author used Anthropic's Claude, a large language model, for
literature search and source verification, drafting and editing of text, derivation checking,
and the production of the replication code and figures. The author reviewed and edited all
content and takes full responsibility for the content of this article.

\paragraph{Competing interests.} The author is Managing Director and Founder of Aegis Compliance
and Strategies O\"U, a compliance and strategy advisory firm, and is preparing a practitioner
book that draws on this line of work.

\paragraph{Funding.} This research received no external funding.

\paragraph{CRediT authorship contribution.} Andreas Bauer: Conceptualization, Methodology,
Formal analysis, Software, Investigation, Writing --- original draft, Writing --- review \&
editing.

\paragraph{Data availability.} See Appendix~\ref{app:data}; replication archive at
\href{https://doi.org/10.6084/m9.figshare.33188046}{doi:10.6084/m9.figshare.33188046} and
companion verification package at
\href{https://doi.org/10.6084/m9.figshare.33190722}{doi:10.6084/m9.figshare.33190722}.

\end{document}